\def\arxivversion{1} %
\documentclass[letterpaper, 10 pt, conference]{ieeeconf} 

\IEEEoverridecommandlockouts 
\usepackage{cite}[sort,compress]
\usepackage{graphicx}
\usepackage{tikz}
\usepackage{circuitikz}
\ctikzset{bipoles/length=1.2cm}
\ctikzset{font=\small}
\usepackage[caption=false,font=footnotesize]{subfig}
\usepackage{amsmath,amssymb,amsfonts}
\usepackage{mathtools, mathrsfs}
\usepackage{xcolor}
\usepackage{stfloats}
\usepackage{bm}
\usepackage{blindtext}

\newcommand{\intv}{\mathcal I}
\newcommand{\R}{\mathbb R}
\newcommand{\B}{\mathfrak B}
\newcommand{\N}{\mathbb N}

\DeclareMathOperator{\col}{col}
\DeclareMathOperator{\rank}{rank}
\newcommand{\diff}{\mathrm d}
\newcommand{\lag}{\bm{l}}
\newcommand{\inputdim}{\bm{m}}
\newcommand{\order}{\bm{n}}
\newcommand{\compl}{\bm{c}}

\newcommand{\mc}[1]{\mathcal{#1}}

\newcommand{\mr}[1]{\mathrm{#1}}
\newcommand{\mb}[1]{\mathbb{#1}}
\newcommand{\ms}[1]{\mathscr{#1}}

\newcommand{\dny}{{n_\mr{y}}}

\newcommand{\dnp}{{n_\mr{p}}}
\newcommand{\dna}{{n_\mr{a}}}

\newcommand{\dnw}{{n_\mr{w}}}

\newcommand{\dnr}{{n_\mr{r}}}
\newcommand{\dnk}{{n_\mr{k}}}
\newcommand{\dnwp}{{n_{\mr{w}'}}}

\makeatletter
\DeclareRobustCommand\vdots{%
  \mathpalette\@vdots{}%
}
\newcommand*{\@vdots}[2]{%
  \sbox0{$#1\cdotp\cdotp\cdotp\m@th$}%
  \sbox2{$#1.\m@th$}%
  \vbox{%
    \dimen@=\wd0 %
    \advance\dimen@ -3\ht2 %
    \kern.5\dimen@
    \dimen@=\wd2 %
    \advance\dimen@ -\ht2 %
    \dimen2=\wd0 %
    \advance\dimen2 -\dimen@
    \vbox to \dimen2{%
      \offinterlineskip
      \copy2 \vfill\copy2 \vfill\copy2 %
    }%
  }%
}
\DeclareRobustCommand\ddots{%
  \mathinner{%
    \mathpalette\@ddots{}%
    \mkern\thinmuskip
  }%
}
\newcommand*{\@ddots}[2]{%
  \sbox0{$#1\cdotp\cdotp\cdotp\m@th$}%
  \sbox2{$#1.\m@th$}%
  \vbox{%
    \dimen@=\wd0 %
    \advance\dimen@ -3\ht2 %
    \kern.5\dimen@
    \dimen@=\wd2 %
    \advance\dimen@ -\ht2 %
    \dimen2=\wd0 %
    \advance\dimen2 -\dimen@
    \vbox to \dimen2{%
      \offinterlineskip
      \hbox{$#1\mathpunct{.}\m@th$}%
      \vfill
      \hbox{$#1\mathpunct{\kern\wd2}\mathpunct{.}\m@th$}%
      \vfill
      \hbox{$#1\mathpunct{\kern\wd2}\mathpunct{\kern\wd2}\mathpunct{.}\m@th$}%
    }%
  }%
}
\makeatother

\usepackage{amsthm}
\theoremstyle{definition}
\newtheorem{definition}{Definition}
\newtheorem{exmp}[definition]{Example}
\newtheorem{remark}[definition]{Remark}
\theoremstyle{remark}

\theoremstyle{plain}
\newtheorem{theorem}{Theorem}
\newtheorem{lemma}[theorem]{Lemma}
\newtheorem{proposition}[theorem]{Proposition}
\newtheorem{corollary}[theorem]{Corollary}
\newtheorem{assumption}{Assumption}

\newenvironment{example}{\begin{exmp}}{\hfill$\blacktriangleleft$\end{exmp}}

\newcommand{\arxivver}[2]{%
  \ifx\arxivversion\undefined%
    #1%
  \else%
    #2%
  \fi%
}

\title{\LARGE \bf
A Continuous-Time Generalization of the LPV Fundamental Lemma
}

\author{Philipp Schmitz and Chris Verhoek%
\thanks{P.~Schmitz is with Technische Universität Ilmenau, Optimization-based Control Group, 98693 Ilmenau, Germany. C.~Verhoek is with the Dept. of Electrical and Systems Engineering at the University of Pennsylvania, Philadelphia, United States, and the Eindhoven University of Technology, Dept. of Electrical Engineering, The Netherlands.}%
\thanks{P.~Schmitz is grateful for the support from the Carl Zeiss Foundation (VerneDCt -- No. 2011640173). C.~Verhoek is gratefully supported by the European project `COVER' under grant No.~101086228.}%
\thanks{Both authors contributed equally. Corresponding author: Philipp Schmitz (e-mail address: \texttt{philipp.schmitz@tu-ilmenau.de}).}%
}

\begin{document}

\maketitle
\arxivver{}{\thispagestyle{plain}
\pagestyle{plain}}

\begin{abstract}
Controller design for physical systems \emph{directly from data} is of tremendous interest in both the industry and academia. However, the research on direct data-driven control has thus far mainly focused on discrete-time linear time-invariant systems. To make the highly non-trivial step to continuous-time nonlinear systems, this paper explores an intermediate route through the framework of \emph{linear parameter-varying} (LPV) systems. LPV systems can be used as a convenient surrogate for nonlinear systems to achieve systematic analysis and controller design. In this work, we generalize the LPV Fundamental Lemma for discrete-time systems towards a class of continuous-time LPV systems, achieving data-driven representations that can be used to design continuous-time LPV controllers.
\end{abstract}
\arxivver{}{
}

\section{Introduction}
Direct data-driven methods can be used to extract system properties or design stabilizing controllers directly from data. Because such approaches are overcoming the costly step of identifying a mathematical description of the system itself, there is a surge for the development of such direct data-driven methods. A corner stone result in many direct data-driven analysis and control methods is the so-called \emph{Fundamental Lemma} by Willems and co-authors~\cite{willems2005note}. This important result uses behavioral system theory for discrete-time \emph{linear time-invariant} (LTI) systems~\cite{yellowbook} to formulate a fully data-based representation of the horizon-$L$ system behavior, using a single {sufficiently informative} trajectory of measurement data. The success of the Fundamental Lemma is evident from the number of data-based (LTI) methods that followed from it, such as application of the data-driven representation in simulation~\cite{markovsky2008data}, performance analysis~\cite{vanwaarde2026time}, and control design~\cite{coulson2019data, berberich2020data}. While the main developments have been focusing on discrete-time systems, novel methods have been proposed recently to obtain similar behavioral data-driven formulations of \emph{continuous-time} LTI systems~\cite{Schmitz2024a, rapisarda2023orthogonal, lopez2022continuous,  othmane2026data, Schmitz2026, Rapisarda2023, Lopez2024}. 

On the other hand, the behavioral direct data-driven methods have been generalized \emph{beyond} LTI dynamics, e.g.,~linear time-varying systems~\cite{nortmann2023direct}, bilinear systems~\cite{markovsky2022data}, {linear parameter-varying} systems~\cite{verhoek2026behavioral, verhoek2021fundamental}, feedback linearizable systems~\cite{alsalti2023data}, convex/affine/conic systems~\cite{padoan2023data}, or nonlinear systems that can represented by a basis function expansion~\cite{lazar2024basis}. While these developments are interesting and important, to our knowledge, extensions for behavioral data-driven representations \emph{beyond continuous-time LTI systems} have not been presented so far. There are works that develop data-driven representations of one-step-ahead predictors of continuous-time input--state systems~\cite{chen2025data, de2023data}. These, however, do not consider input--output data and assume that the derivative of the state can be measured, which is generally a restrictive assumption in practice.

{Extending the fundamental lemma to continuous time is not merely a matter of generality. Continuous-time data-driven methods apply directly to the dynamics of physical systems without a discretization step. This is desirable, for instance, in power electronics, where analog circuits are used to implement fast inner-loop control at low cost. Moreover, from a theoretical standpoint, we see that the discretization step of an LPV system is generally only approximate, since the dynamics depend on the scheduling trajectory within each sampling interval. A continuous-time formulation avoids this approximation.}

To make the highly non-trivial step from data-driven representations for discrete-time LTI systems to data-driven representations continuous-time nonlinear systems, we explore an intermediate route through the framework of \emph{continuous-time linear parameter-varying (LPV) systems}. The LPV framework~\cite{toth2026modeling} considers systems whose behavior is defined by a linear, dynamic relationship. The system parameters that characterize this relationship are varying along a measurable \emph{scheduling signal}, denoted by~$p$. LPV systems are often used in practice as surrogates of nonlinear dynamics, where the scheduling signal captures the nonlinearities and exogenous effects~\cite{rolandbook, toth2026modeling}. This allows for efficient and systematic analysis and control design of nonlinear systems. Behavioral data-driven representations have been developed for discrete-time LPV systems~\cite{verhoek2026behavioral, verhoek2021fundamental}. However, methods for constructing direct data-driven LPV representations of \emph{continuous-time LPV systems} are yet to be developed.

To fill this gap in the literature, our contribution is the development of data-driven representations for a special class of continuous-time LPV systems. Specifically, LPV systems that can be represented by a kernel representation that has a, what we call, \emph{Leibniz-affine} scheduling dependence. Our contributions are two-fold: First, we present a reformulation of the continuous-time Fundamental Lemma for LTI systems as an \emph{identifiability condition}, resulting in conditions on the input--output data for the construction of a data-driven representation. Second, we generalize this first contribution for the class of LPV systems with Leibniz-affine scheduling dependence through the use of a temporary LTI embedding. To our knowledge, this is the first result to formulate continuous-time behavioral data-driven representations from input--output data beyond the class of LTI systems.

The remainder is structured as follows. We formalize the setting and the considered problem in Section~\ref{s:prob}. {Section~\ref{s:lpv_behav} details LPV behaviors, their various representations and key structural quantities.} Inspired by~\cite{verhoek2026behavioral}, Section~\ref{s:ltiemb} proposes an LTI embedding technique for the considered class of LPV systems, which allows us to formulate {and prove} a continuous-time LPV Fundamental Lemma in Section~\ref{s:fl}. We draw conclusions and provide possible future research directions in Section~\ref{s:conc}. {All proofs can be found \arxivver{in~\cite{schmitzverhoek2026}.}{in the appendices.}}

\subsubsection*{Notation}
$\R$ and $\N$ denote the real and natural numbers, respectively. For compatible matrices $A_1,\dots, A_m$ we define $\col(A_1,\dots,A_m) := \begin{bmatrix} A_1^\top & \dots & A_m^\top \end{bmatrix}{}^\top$. The Kronecker product of two matrices $A$ and $B$ is denoted by $A\otimes B$. Both notations extend to matrix-valued functions. Let $\intv\subset \R$ be an open interval. The $k$\textsuperscript{th} order Sobolev space induced by the space $L^q(\intv, \R^d)$, with $q\in \N\cup \{\infty\}$, is denoted by $W^{k,q}(\intv,\R^d)$. The space of infinitely differentiable, compactly supported functions from $\intv$ to $\R^d$ is denoted by $\mathcal C^\infty_\mathrm{c}(\intv, \R^d)$. {Derivatives of a function $f$ are denoted by $\tfrac{\diff^{k}}{\diff t^{k}} f = f^{(k)}$, $k\in\N$. In this work, $f'$ denotes an augmented variable and is not to be confused with the derivative $f^{(1)}$. Finally, given two sets $\mathbb W$ and $\mathbb T$, we use the notation $\mathbb W^{\mathbb T}$ for the set of functions from $\mathbb T$ to $\mathbb W$.}

\section{Problem description}\label{s:prob}
LPV dynamical systems are defined by the quadruple $\Sigma:=(\mb{T}, \mb{W},\mb{P},\B)$, with time axis~$\mb{T}$, signal space~$\mb{W}\subseteq\R^\dnw$, scheduling space~$\mb{P}\subseteq\mb{R}^\dnp$, and behavior $\B\subseteq(\mb{W}\times\mb{P})^\mb{T}$. The behavior~$\B$ is linear in the sense that for a given~$p$, any~$(w,p),(\tilde{w},p)\in\B$, and $a, b\in\mb{R}$, $(a  w+b\tilde{w},p)\in\B$. Given $p\in{\mathbb P^\mathbb T}$, the set of all trajectories $w$ compatible with~$p$ is defined by $\B_p:=\{w\,|\, (w,p)\in\B\}$, which is a linear subspace of~$\mathbb W^\mathbb T$. {Assume throughout that~$\mb{P}$ is convex with a non-empty interior.}

To handle non-commutative differential operations in the parametrization of LPV representations, for example,~$\frac{\diff}{\diff t}p(t) w(t) \left(\neq p(t) \frac{\diff}{\diff t} w(t)\right)$, we introduce the $\diamond$-operator from~\cite{rolandbook}. For a coefficient function $r:\mb{P}^\mb{T}\to\mb{R}^{\cdot\times\cdot}$ and scheduling signal~$p\in\mb{P}^\mb{T}$, the notation~$r\diamond p$ allows us to consider scheduling-dependent parameterizations that are dependent on a finite number of derivatives of~$p$, i.e., 
\[(r\diamond p)(t) = r(p, \tfrac{\diff p}{\diff t}, \dots, \tfrac{\diff^n p}{\diff t^n}).\]
See~\cite[Chap.~3]{rolandbook} for a more detailed explanation. Using this operator, we can compactly write polynomials in the indeterminate~$\xi$, whose coefficients are dependent on (derivatives of) the scheduling variable, e.g.,~$R(\xi)\diamond p=(r_1\diamond p)\xi^2 + (r_2\diamond p)\xi = (3p + 2\tfrac{\diff p}{\diff t})\xi^2 + (\sin(p) + \tfrac{\diff^2 p}{\diff t^2})\xi$. We are now ready to introduce LPV kernel representations, which are used to characterize the behavior~$\B$. 

We consider continuous-time LPV systems that can be described by the kernel representation
\begin{equation}
   \label{eq:kernel-gen}
    (R(\tfrac{\diff}{\diff t})\diamond p)w = 0
\end{equation}
on an open finite time interval $\intv$, i.e., $\mb{T}=\intv$. Here, $p:\intv\to\R^{\dnp}$ is the scheduling signal and $w:\intv\to\R^{\dnw}$ is the system trajectory. The dynamics of the LPV system are characterized by the polynomial matrix function $R(\xi)$ in the indeterminant~$\xi$, where~$\xi^n$ corresponds to the $n$\textsuperscript{th} time derivative. 

We restrict ourselves to a specific type of dynamic, functional scheduling dependence for $R$. We assume that the coefficients of~$R$ are affinely dependent on $R$, and follow a Leibniz-like dynamic dependence. Hence,~$R$ is such that we can write~\eqref{eq:kernel-gen} as
\begin{equation}\label{eq:kernel}
    (R(\tfrac{\diff}{\diff t})\diamond p)w =\sum_{i=0}^{\dnr}\Big( R_{i,0} \frac{\diff^i w}{\diff t^i} + \sum_{j = 1}^{\dnp} R_{i,j} \frac{\diff^i (p_j w)}{\diff t^i}\Big)=0,
\end{equation}
where $R_{i,j}\in\R^{\dnk \times \dnw}$. We call this dependency structure \emph{Leibniz-affine}, since the dynamic and functional scheduling dependence arises from derivatives of only the product of the manifest and scheduling signals, excluding any isolated mixed derivative terms. While restrictive compared to the general framework presented in~\cite{rolandbook}, the class of Leibniz-affine captures a relatively large class of physical systems, e.g., RC circuits with varying capacitance (see Example~\ref{example}) or mass-spring-damper systems with velocity-dependent damping coefficients. Moreover, as shown in
\arxivver{\cite[App.~C]{schmitzverhoek2026},}{Appendix~\ref{app:proof:lem:complexity},}
systems with a Leibniz-affine scheduling dependence admit a direct state-space realization with static affine dependence, allowing for convex, linear matrix inequality-based analysis and control synthesis.
{\begin{example}\label{example}
\begin{figure}
    \centering
    \begin{tikzpicture}[european, scale=0.7]
    \node[ocirc] (A) at (0,3) {};
    \node[ocirc] (B) at (0,0) {};
    \draw (A) to[open, v=$u(t)$, voltage=straight] (B);
    \draw (A.east) to[R, l=$R_0$, -*] (3,3) -- (5,3) node[ocirc] (C) {};
    \draw (3,3) to[vC, l_=$C(t)$, *-*] (3,0);
    \draw (B.east) -- (3,0) -- (5,0) node[ocirc] (D) {};
    \draw (C) to[open, v^=$y(t)$, voltage=straight] (D);
    \end{tikzpicture}
    \caption{RC-circuit with variable capacitance.}
    \label{fig:RC}
\end{figure}
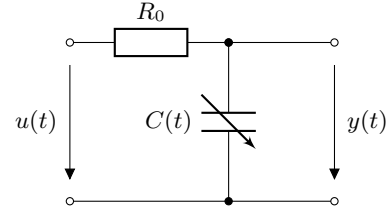
    An example of a Leibniz-affine LPV system of the form~\eqref{eq:kernel} is the RC-circuit in Fig.~\ref{fig:RC}. Define the input and output voltage $u$, $y$, constant resistance $R_0$ and a variable capacitance $C(t)$. Let $C(t):=C_0 + C_1p(t)$, with~$p$ the scheduling variable. Then, Kirchhoff's voltage law gives:
    \[R_0C_0y^{(1)} + R_0C_1(py)^{(1)} + y - u = 0.\]
    With $w=\col(u,y)$, we thus obtain the form~\eqref{eq:kernel}:
    \[  \begin{bsmallmatrix} -1 & 1 \end{bsmallmatrix}w + \begin{bsmallmatrix} 0 & R_0C_0 \end{bsmallmatrix}w^{(1)} + \begin{bsmallmatrix} 0 & 0 \end{bsmallmatrix}pw + \begin{bsmallmatrix} 0 & R_0C_1 \end{bsmallmatrix}(pw)^{(1)} = 0, \]
    which has Leibniz-affine scheduling dependence.
\end{example}
In this paper, we are interested in {LPV behaviors} that are characterized by~\eqref{eq:kernel}.
We investigate the conditions under which these can be fully characterized on the interval~$\intv$ by only measurement data. %

Before we present the solution to this problem, we need to discuss the solutions and representation variations of~\eqref{eq:kernel}.
\section{Leibniz-affine LPV behaviors \& representations}
\label{s:lpv_behav}

{In this section, we define the behavior associated with~\eqref{eq:kernel} as a function space. Further, we recall structural quantities of behaviors and alternative representation forms.
\subsection{Kernel Representations via Weak Solutions}
We characterize the LPV behavior of~\eqref{eq:kernel}, i.e., the set of all valid pairs of trajectories and scheduling signals of~\eqref{eq:kernel}, in a weak sense. Specifically, we further define~$\B$ as:}
\begin{equation}\label{eq:behavior}
   \B := \left\{(w,p)\in L^2(\intv, \R^{\dnw}) \times L^\infty(\intv, \R^{\dnp})\,\middle|\,\begin{aligned}
    &\text{\eqref{eq:kernel} holds}\\[-.2em]
    &\text{weakly}
   \end{aligned} \right\}.
\end{equation}
A pair~$(w,p)\in L^2(\intv, \R^{\dnw}) \times L^\infty(\intv, \R^{\dnp})$ is called a weak solution to~\eqref{eq:kernel} if for all $\phi\in\mathcal C_{\mathrm c}^\infty(\intv,\R^{\dnk})$, it holds that 
\begin{equation}
   \label{eq:weak}
   \int_\intv \bigg(\sum_{k=0}^{\dnr} (-1)^{k} \Big(w R_{k,0}^\top  + \sum_{j = 1}^{\dnp}  p_jw R_{k,j}^\top \Big) \tfrac{\diff^k \phi}{\diff t^k}\bigg) \diff t = 0.
\end{equation}
{Given a pair $(w,p)$ of infinitely differentiable, bounded functions, the equivalence of \eqref{eq:kernel} and \eqref{eq:weak} follows immediately from integration by parts of~\eqref{eq:weak}, and the fact that the space of test functions $\mathcal C_\mathrm{c}^\infty(\intv, \R^{\dnk})$ is dense in $L^2(\intv, \R^{\dnk})$.}
{This allows us to conclude the following:}
\begin{lemma}
\label{lem:closed}
      The behavior $\B$ defined as in~\eqref{eq:behavior} is a closed subset of $L^2(\intv, \R^{\dnw}) \times L^\infty(\intv, \R^{\dnp})$.
\end{lemma}
Choosing $\B$ as a subset of $L^2(\intv,\R^{\dnw})\times L^\infty(\intv, \R^\dnp)$ is advantageous because, by Lemma~\ref{lem:closed}, it ensures that $\B$ is a Banach space, providing a robust framework of analytical treatment. Alternatively, the behavior can be defined more broadly using spaces of locally integrable functions as discussed in \cite{rolandbook}.

{In the remainder, when we refer to~$\B$, we always mean the LPV behavior as defined in~\eqref{eq:behavior}. Next, we will discuss some characteristic structural quantities of $\B$, referred to as the \emph{complexity}.}

\subsection{Complexity of LPV Behaviors}
The {complexity} of an LPV behavior~$\B$ is characterized by the triple of integers
\begin{equation}
   \compl(\B) = (\order(\B), \lag(\B), \inputdim(\B)).
\end{equation}
The integer $\order(\B)$ is the \emph{order} of~$\B$, corresponding to the smallest possible state dimension of a state-space representation with \emph{static scheduling dependence} whose manifest behavior coincides with~$\B$. The integer~$\lag(\B)$ is the \emph{lag} of~$\B$, and is the smallest possible $\dnr$ over all kernel representations of the form~\eqref{eq:kernel} that represent~$\B$. The integer~$\inputdim(\B)$ is the \emph{input dimension}, which is defined as follows: If for a partitioning $w=\col(u,y)$, $u$ is such that it is maximally free, i.e., none of the components of~$y$ can be chosen freely for all~$p$, then~$\inputdim(\B)$ corresponds to the dimension of $u$. These quantities are all well-defined in~\cite[Chap.~3]{rolandbook}, and are closely linked to minimality properties of the representations of~$\B$, e.g., a minimal~\eqref{eq:kernel} has $\dnk = \dnw-\inputdim(\B)$.

\subsection{Input--Output Representations}
When considering a partitioning of~$w$ in terms of inputs~$u$ and outputs~$y$, it is natural to introduce input--output of~\eqref{eq:kernel}. For input--output representations, we pick a partitioning\footnote{In classical behavioral theory~\cite{yellowbook}, the partition of~$w$ is generally defined by a full-rank permutation matrix~$\Pi$, such that~$w=\Pi\col(u,y)$. Throughout, we take w.l.o.g.~$\Pi=I$.} $w=\col(u,y)$, such that~$u:\intv\mapsto\R^{\inputdim(\B)}$ is maximally free and $y:\intv\mapsto\R^{\dny}$, where $\dny:=\dnw-\inputdim(\B)$. We then split up $R_{i,j} = \begin{bmatrix} -B_{i,j} & A_{i,j} \end{bmatrix}$ accordingly, where $A_{i,j}\in\R^{\dnk\times\dny}$ and $B_{i,j}\in\R^{\dnk\times\inputdim(\B)}$,  such that we obtain
\begin{multline}\label{eq:io}
    \textstyle \sum_{i=0}^{\dnr}\Big( A_{i,0} \tfrac{\diff^i y}{\diff t^i} + \sum_{j = 1}^{\dnp} A_{i,j} \tfrac{\diff^i (p_j y)}{\diff t^i}\Big) \\
    \textstyle=\sum_{i=0}^{\dnr}\Big( B_{i,0} \tfrac{\diff^i u}{\diff t^i} + \sum_{j = 1}^{\dnp} B_{i,j} \tfrac{\diff^i (p_j u)}{\diff t^i}\Big),
\end{multline}
which is unique up to permutations. Note that we can write~\eqref{eq:io} in the short form $(A(\frac{\diff}{\diff t})\diamond p)\,y=(B(\frac{\diff}{\diff t})\diamond p)\,u$ if desired. We assume the following throughout the paper when we consider the input--output realization~\eqref{eq:io}:\footnote{While left-coprimeness of $A(\xi)\diamond p$ and $B(\xi)\diamond p$ can always be achieved, requiring $A_{\dnr,0}=I$ is restrictive. We believe that this assumption can be lifted without affecting the validity of the results that depend on it. However, proving the results without this restriction on $A_{\dnr,0}$ will require extensive and deep discussions on bridging LPV and LTI realization theory in the continuous-time behavioral setting. As this is not the focus of this paper, we leave this to future work.}
\begin{assumption}\label{ass:coprime-etc}
    The polynomials $A(\xi)\diamond p$ and $B(\xi)\diamond p$ defining~\eqref{eq:io} are generically coprime in $p$, i.e., their greatest common divisor is one for generic $p\in\mathbb{P}^\mathbb{T}$. Moreover,
    the leading coefficient $A_{\dnr,0}=I_{\dny}$ and $A_{\dnr,j}=0$ for $j>0$.
\end{assumption}
\begin{remark}
The special case of~\eqref{eq:io} where $\dnr=1$, $A_{1,0}=I$, and $A_{1,j}=0$ and $B_{1,j-1}=0$ for $j>0$ yields the well-known LPV input--state realization with affine scheduling dependence:
\begin{equation*}
    \tfrac{\diff }{\diff t} y = -A_{0,0}y - \sum_{j = 1}^{\dnp} A_{0,j} p_j y + B_{0,0}u + \sum_{j = 1}^{\dnp} B_{0,j} p_j u.
\end{equation*}
\end{remark}

With the formal definition of~$\B$, the complexity, and input--output representations, we now introduce LTI embeddings of LPV behaviors. These embeddings allow us to formulate our main result: data-driven representations of continuous-time LPV systems.

\section{LTI embeddings of LPV behaviors}\label{s:ltiemb}
Inspired by the discrete-time LPV Fundamental Lemma~\cite{verhoek2026behavioral}, we achieve our generalization to the continuous-time case through a (temporary) LTI embedding of~$\B$. The embedding is characterized by auxiliary input signals. Given $(w,p)$, we introduce
\begin{equation}
   \label{eq:augment}
   w' := \begin{bmatrix}
      w\\p\otimes w
   \end{bmatrix} = \begin{bmatrix}
      w\\p_1w\\\vdots\\p_{\dnp}w
   \end{bmatrix}.
\end{equation}
If we now \emph{temporarily disregard} the dependence on~$p$ and consider $p_1w, \dots,p_\dnp w$ as auxiliary input signals, i.e., they are considered to be free signals, we can define the augmented behavior $\B'$ as
\begin{equation}
   \B' = \left\{w'\in L^2(\intv, \R^{\dnw(1+\dnp)})\,\middle|\, \begin{aligned} & R'(\tfrac{\diff}{\diff t})w'=0\\ &\text{holds weakly}\end{aligned}\right\},
\end{equation}
with $R'_i = \begin{bmatrix} R_{i,0} & \cdots & R_{i,\dnp}\end{bmatrix}\in\R^{\dnk\times\dnw(1+\dnp)}$, and where $R'(\tfrac{\diff}{\diff t})w'=0$ is to understood in again in the weak sense. That is, we say that~$w'$ is a weak solution to $R'(\tfrac{\diff}{\diff t})w'=0$ if for all $\phi\in\mathcal C_{\mathrm c}^\infty(\intv,\R^\dnk)$, it holds that
\begin{equation}
   \label{eq:weak2}
   \int_\intv \sum_{k=0}^{\dnr} (-1)^{k} w' (R_{k}')^\top \tfrac{\diff^k \phi}{\diff t^k}\ \diff t = 0.
\end{equation}
The following lemma establishes the formal relationship between $\B$ and $\B'$.
\begin{lemma}
   \label{lem:embedding}
   {A pair $(w,p)\in L^2(\intv, \R^{\dnw}) \times L^\infty(\intv, \R^{\dnp})$  belongs to the LPV behavior $\B$ if and only there exists $w' \in \B'$ satisfying $w'=\col(w,p\otimes w)$}.
\end{lemma}
By identifying each pair $(w,p)\in\B$ with its augmented trajectory $w'=\col(w,p\otimes w)\in\B'$, the LPV behavior $\B$ can be viewed as a subset of the augmented LTI behavior~$\B'$, i.e., with a slight abuse of notation we conclude $\B\subseteq\B'$.
However, because we assumed that $p_1w, \dots,p_\dnp w$ are free auxiliary input signals, we in fact have $\B\subset\B'$. Hence, 
the augmented behavior $\B'$ defined by the LTI embedding is an \emph{over-approximation} of~$\B$, which we illustrated in Figure~\ref{fig:behaviors}. 
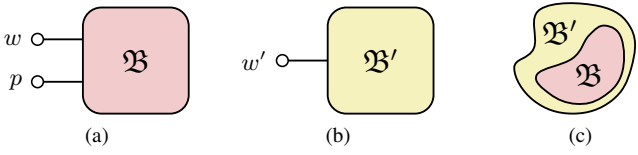
\begin{figure}[t]
  \centering
  \subfloat[]{%
    \begin{tikzpicture}
  \definecolor{mypink}{rgb}{0.95,0.80,0.80}
  \draw[fill=mypink, draw=black, rounded corners=7pt, line width=0.7pt]
    (-0.7,-0.7) rectangle (0.7,0.7);
  \node at (0,0) {\large$\mathfrak{B}$};
  \draw[line width=0.7pt] (-1.3,0.28) -- (-0.7,0.28);
  \filldraw[draw=black,fill=white,line width=0.6pt] (-1.3,0.28) circle (2.2pt);
  \node[anchor=east,font=\small] at (-1.38,0.28) {$w$};
  \draw[line width=0.7pt] (-1.3,-0.28) -- (-0.7,-0.28);
  \filldraw[draw=black,fill=white,line width=0.6pt] (-1.3,-0.28) circle (2.2pt);
  \node[anchor=east,font=\small] at (-1.38,-0.28) {$p$};
\end{tikzpicture}%
  }\hfill
  \subfloat[]{%
    \begin{tikzpicture}
  \definecolor{myyellow}{rgb}{0.96,0.95,0.74}
  \draw[fill=myyellow, draw=black, rounded corners=7pt, line width=0.7pt]
    (-0.7,-0.7) rectangle (0.7,0.7);
  \node at (0,0) {\large$\mathfrak{B}'$};
  \draw[line width=0.7pt] (-1.3,0) -- (-0.7,0);
  \filldraw[draw=black,fill=white,line width=0.6pt] (-1.3,0) circle (2.2pt);
  \node[anchor=east,font=\small] at (-1.38,0) {$w'$};
\end{tikzpicture}%
  }\hfill\hspace{10pt}
  \subfloat[]{%
    \begin{tikzpicture}[scale=0.43]
\definecolor{mypink}{rgb}{0.95,0.80,0.80}
\definecolor{myyellow}{rgb}{0.96,0.95,0.74}
\fill[myyellow, draw=black, line width=0.7]
  plot[smooth cycle, tension=0.7] coordinates {
    (0,0)
    (.2,1.2)
    (2.0,1.3)
    (3,0.6)
    (3,-0.6)
    (2.4,-1.6)
    (1.2,-2)
    (0,-1.6)
    (-0.6,-0.6)
  };
\begin{scope}[shift={(0.1,-0.3)}]
\fill[mypink, draw=black, line width=0.7]
  plot[smooth cycle, tension=0.7] coordinates {
    (1.1,0.0)
    (1.6,0.7)
    (2.1,1.0)
    (2.6,0.6)
    (2.6,-0.5)
    (2,-1.2)
    (1.2,-1.4)
    (0.4,-1.1)
    (0,-0.5)
  };
  \end{scope}

\node at (0.8,0.5) {\large$\mathfrak{B}'$};
\node at (1.7,-0.8) {\large$\mathfrak{B}$};

\end{tikzpicture}%
  }
  \caption{(a)~The LPV behavior $\B$ with scheduling signal $p$ and manifest variable $w$. (b)~The augmented LTI behavior $\B'$ with variable $w'$, composed of $w$ augmented with $\dnp\dnw$ auxiliary input signals. (c)~Illustration of the embedding $\B\subset\B'$. (Figure adopted from~\cite{verhoek2026behavioral}).}
  \label{fig:behaviors}
\end{figure}

Similar to the LPV behavior~$\B$, the complexity of the LTI behavior $\B'$ (see, e.g.,~\cite{markovsky2022identifiability}) is characterized by the triple
\begin{equation}
   \compl(\B') = (\order(\B'), \lag(\B'), \inputdim(\B')),
\end{equation}
where $\order(\B')$ is the McMillan degree of $R'$ or, equivalently, the state dimension of a minimal state space representation, $\lag(\B')$ is the lag of $\B'$, and $\inputdim(\B')$ is the input dimension of $\B'$, see \cite{yellowbook, markovsky2022identifiability}.
With the LTI embedding of~\eqref{eq:kernel} characterized by~\eqref{eq:augment}, the values that define the complexity of the augmented behavior~$\B'$ are \emph{inherited} from $\compl(\B)$. We show this in the following result:

\begin{lemma}
   \label{lem:complexity}
    Let Assumption~\ref{ass:coprime-etc} hold. The relation between the complexities of~$\B$ and~$\B'$ is
        $\order(\B)=\order(\B')$, $\lag(\B) = \lag(\B')$ and $\inputdim(\B) + \dnp \dnw=\inputdim(\B')$.
\end{lemma}

With the LTI embedding and its complexity~$\compl(\B')$ defined, we now present the main results of our paper.

\section{A Fundamental Lemma}\label{s:fl}
In this section, we present a data-driven representation for continuous-time LPV systems with behavior~$\B$. We build on the continuous-time Fundamental Lemma for LTI systems in~\cite{Schmitz2024a} and adopt a reformulation tailored to our setting. Specifically, we first present a continuous-time analog of the discrete-time identifiability condition, as presented in~\cite[Cor.~21]{markovsky2022identifiability}.
\subsection{An LTI Fundamental Lemma}
The continuous-time Fundamental Lemma in~\cite{Schmitz2024a} is formulated as an input-design result. This provides \emph{sufficient, but a priori,} conditions on the input data to determine whether the data is able to represent the behavior. In the following, we provide \emph{necessary and sufficient, but a posteriori}, conditions on the input--output data to determine whether the data can represent the behavior. This provides us with the tools to formulate a continuous-time LPV Fundamental Lemma, and draw parallels with the discrete-time LPV Fundamental Lemma~\cite[Thm.~1]{verhoek2026behavioral}.

\begin{proposition}
   \label{prop:FL_LTI}
   Define $\dnwp := \dnw(1+\dnp)$, and let $L\geq\lag(\B')+1$. Let $w'\in\B'\cap W^{L-1,2}(\intv,\R^{L\dnwp\times L\dnwp})$ and define the Gramian matrix
   \begin{equation}
      \label{eq:gamma}
      \Gamma = \int_\intv \mathcal W(t)\mathcal W(t)^\top\,\diff t,\quad \mathcal W=\begin{bmatrix}
         w'\\\vdots\\(w')^{(L-1)}
      \end{bmatrix}.
   \end{equation}
   Then, the following statements are equivalent:
   \begin{enumerate}
       \item[(i)] $\rank(\Gamma) = L\inputdim(\B') + \order(\B')$,
       \item[(ii)] $\bar w'\in W^{L-1,2}(\intv,\R^{\dnwp})$ is a trajectory of $\B'$ if and only if there exists $g\in L^2(\intv,\R^{L\dnwp})$ such that
   \begin{equation}
      \label{eq:FL_LTI}
      \begin{bmatrix}
         \bar w'\\
         \vdots\\
         (\bar w')^{(L-1)}
      \end{bmatrix} = \Gamma g.
   \end{equation}
   \end{enumerate}
\end{proposition}

{Proposition~\ref{prop:FL_LTI} establishes a necessary and sufficient condition on the measured input--output data in the Gramian matrix~$\Gamma$, guaranteeing that the LTI behavior~$\B'$ with complexity~$\compl(\B')$ can be represented by the data matrix~$\Gamma$}. As we mentioned earlier, this result presents a continuous-time analog of the \emph{identifiability} result of~\cite{markovsky2022identifiability}. This connects to the fact that Proposition~\ref{prop:FL_LTI} does \emph{not} require any controllability assumption on the behavior $\B'$, contrary to the continuous-time Fundamental Lemma in~\cite{Schmitz2024a}.
With this observation, we also want to addresses the slight misuse of terminology in this work: in the data-driven literature the term ``Fundamental Lemma'' has been used interchangeably for \emph{both} the input-design condition, originally proposed in~\cite{willems2005note}, and the resulting data-driven representation itself, which has been used extensively for analysis and control. Although, we do not derive an input-design condition, we will stick to the terminology of ``Fundamental Lemma.''

\subsection{An LPV Fundamental Lemma}
{We can now generalize Proposition~\ref{prop:FL_LTI} for the considered class of LPV systems, which results in the formulation of an LPV Fundamental Lemma.}
\begin{theorem}
   \label{thm:FL_LPV}
   Let Assumption~\ref{ass:coprime-etc} hold and define $\dnwp := \dnw(1+\dnp)$. Consider $L\in\N$, with $L\geq \lag(\B)+1$,  and let $(w,p)\in\B$ such that the Gramian matrix~$\Gamma$ in \eqref{eq:gamma} for $w'=\col(w,p\otimes w)$ satisfies 
   \begin{equation}\label{eq:ranklpv}
      \rank(\Gamma) = L(\inputdim(\B)+\dnp\dnw) + \order(\B).
   \end{equation}
   Then, $(\bar w,\bar p)\in W^{L-1,2}(\intv, \R^{\dnw})\times W^{L-1,\infty}(\intv,\R^{\dnp})$ is a trajectory of $\B$ if and only if there exists $g\in L^2(\intv,\R^{L\dnwp})$ such that
   \begin{equation}\label{eq:FL_LPV}
      \begin{bmatrix}
         \bar w\\
         \bar p\otimes \bar w\\
         \vdots\\
         \bar w^{(L-1)}\\
         (\bar p\otimes \bar w)^{(L-1)}
      \end{bmatrix} = \Gamma g.
   \end{equation}
\end{theorem}

Other than the continuous-time Fundamental Lemma for LTI systems \cite[Thm.~22]{Schmitz2024a}, Theorem~\ref{thm:FL_LPV} is formulated based on a rank condition on the full input--output data trajectory instead of an input-design criteria. Hence, it provides a condition on the data that is verified a posteriori. The discrete-time identifiability results for LTI and LPV systems in~\cite{markovsky2022identifiability, verhoek2026behavioral} provide necessary and sufficient conditions on the input--output data to obtain a data-driven representation. We consider showing necessity of the rank condition~\eqref{eq:ranklpv} an important future research direction (similarly for the rank condition in Proposition~\ref{prop:FL_LTI}).

{The following analysis explores the nuances of~\eqref{eq:FL_LPV}, revealing an insightful structure that is not immediately apparent.
Given a scheduling signal $\bar p\in W^{L-1, \infty}(\intv, \R^{\dnp})$, we define the scheduling-dependent {matrix-valued function with lower block-triangular structure}
\begin{equation*}
   \bar{\mathscr P_L} := \begin{bsmallmatrix} 
\binom{0}{0}\,\bar p & 0 & \cdots & 0\\[0.5em]
\binom{1}{1}\,\bar p^{(1)} & \binom{1}{0}\,\bar p & \ddots & \vdots\\[0.5em]
\vdots & \vdots & \ddots & 0 \\[0.5em]
\binom{L-1}{L-1}\,\bar p^{(L-1)} & \binom{L-1}{L-2}\,\bar p^{(L-2)} & \cdots &\binom{L-1}{0}\,\bar p %
\end{bsmallmatrix}\otimes I_\dnw,
\end{equation*}
with the binomial coefficient $\binom{n}{k}$ for nonnegative integers~$n, k$. Note that the entries of $\bar{\mathscr P_L}$ are signals, i.e., $\bar{\mathscr P_L}$ is a function of time.

The definition of~$\bar{\mathscr P_L}$ allows us to formulate the following insightful result:
\begin{corollary}\label{cor:ofthmfl_lpv}
    Let the assumptions of Theorem~\ref{thm:FL_LPV} hold and consider the partitioning of $\Gamma$,
    \begin{equation}
        \Gamma=\begin{bmatrix}
            \Gamma_{w}\\
            \Gamma_{pw}\\
            \vdots
            \\
            \Gamma_{w^{(L-1)}}\\
            \Gamma_{(pw)^{(L-1)}}
        \end{bmatrix}
    \end{equation}
    with $\Gamma_w,\Gamma_{w^{(j)}} \in \R^{\dnw \times \dnwp}$, $\Gamma_{pw}, \Gamma_{(pw)^{(j)}}\in\R^{\dnw\dnp\times \dnwp}$, $j=1,\dots, L-1$.
    Then, $(\bar w,\bar p)\in W^{L-1,2}(\intv, \R^{\dnw})\times W^{L-1,\infty}(\intv,\R^{\dnp})$ is a trajectory of $\B$ if and only if there exists a $g\in L^2(\intv,\R^{L\dnwp})$ with
    \begin{equation}
    \label{eq:rep}
        \bar w = \Gamma_{w}g, \quad (\bar p\otimes \bar w) = \Gamma_{pw} g
    \end{equation}
    satisfying the differential-algebraic equation
    \begin{equation}
    \label{eq:DAE}
       \frac{\diff}{\diff t} \left(\begin{bmatrix}
            \Gamma_{w}\\\vdots\\\Gamma_{w^{(L-2)}}
        \end{bmatrix}g\right) = \begin{bmatrix}
            \Gamma_{w^{(1)}}\\\vdots\\\Gamma_{w^{(L-1)}}
        \end{bmatrix} g
    \end{equation}
    and the compatibility condition
    \begin{equation}
    \label{eq:compat}
        \begin{bmatrix} \Gamma_{pw}\\\vdots\\\Gamma_{(pw)^{(L-1)}} 
        \end{bmatrix} g = \bar{\mathscr P_L} \begin{bmatrix}
            \Gamma_{w}\\\vdots\\\Gamma_{w^{(L-1)}}
        \end{bmatrix} g.
    \end{equation}
\end{corollary}

This result allows us to draw interesting parallels with the discrete-time case studied in \cite{verhoek2021fundamental}, which we collect in the following remarks. 
\begin{remark}
    {As a consequence of the embedding $\B$ into the continuous-time LTI behavior $\B'$, the data-driven description of trajectories of system~\eqref{eq:kernel} via \eqref{eq:rep}--\eqref{eq:compat} is not restricted to the specific time interval $\intv$. Given a different time interval~$\tilde \intv$ and the corresponding behavior~$\tilde\B$ generated by the dynamics~\eqref{eq:kernel} on $\tilde\intv$, the same data matrix $\Gamma$ obtained from a trajectory on $\intv$ can be used to describe trajectories on $\tilde\intv$, cf.\ \cite[Rem.~7]{Schmitz2026}, where the parameter function $g$ is now defined on $\tilde\intv$.}
\end{remark}

\begin{remark}
    The discrete-time analog to~$\bar{\ms{P}}_L$ is a matrix~$\bar{\mc{P}}_L$ with block-\emph{diagonal} structure, containing the scheduling sequence, see \cite[Eq.~(24)]{verhoek2026behavioral}. For a fixed scheduling sequence~$\bar{p}$, the associated~$\bar{\mc{P}}_L$ is a constant matrix and allows for a characterization of the set of all $\bar p$-compatible input--output trajectories, cf.\ $\B_{\bar p}$ in Section~\ref{s:prob}. %
    In contrast,~$\bar{\ms{P}}_L$ is a matrix-valued function of time with lower block-triangular structure populated with derivatives of~$\bar p$ below the diagonal. %
    An analog description of the $\bar p$-compatible trajectories in $\B_{\bar p}$ for the continuous-time case %
    would incur that we project the solution space of the DAE \eqref{eq:DAE} %
    onto the space of all functions that satisfy~\eqref{eq:compat}. In the discrete-time case, this kind of projection is achieved using straight-forward linear algebra (projection of the image of a Hankel matrix onto the basis of a null space, which can easily be computed for a given~${\bar p}$, see \cite[Eqs.~(26)--(28)]{verhoek2026behavioral}). The resulting subspace has dimension $\order(\B) + L \inputdim(\B)$, emphasizing the linearity of~$\B$ along a given scheduling trajectory. Achieving a similar condition for the continuous-time case is an additional important topic for future work.
\end{remark}
}

\begin{remark}
    {Finally, we acknowledge that a major challenge for the application of Theorem~\ref{thm:FL_LPV} is, that the data-matrix~$\Gamma$ involves higher-order derivatives of the input, output and scheduling signals, which are typically sensitive to noise. To address this issue, an efficient and robust technique using algebraic differentiators is proposed in~\cite{othmane2026data}.}
\end{remark}

\section{Conclusion and future works}\label{s:conc}
We generalized the LPV Fundamental Lemma for continuous-time LPV systems that can be represented by a kernel representation with a Leibniz-affine scheduling dependence. {To the best of our knowledge, this result provides the first behavioral} data-driven representation for continuous-time systems, \emph{beyond the class of LTI systems}. To achieve this result, we first formulated an identifiability condition for continuous-time LTI systems, which followed from reformulating the continuous-time LTI Fundamental Lemma. A temporarily introduced LTI embedding for the considered class of LPV systems allowed us to generalize the LTI identifiability condition for LPV systems. 

{Future work will focus on formulating necessary conditions for the identifiability results and extending the presented results to a broader class of continuous-time LPV systems. Furthermore, to make our results actionable and apply them in, e.g., data-driven optimal control or data-driven simulation, robust numerical methods must be developed to efficiently and reliably solve the differential-algebraic equation~\eqref{eq:DAE} subject to the compatibility condition~\eqref{eq:compat}.} {These methods will also enable the demonstration of our methods in realistic, numerical examples.}

\section*{Acknowledgements}

We thank the anonymous reviewers for their comments, in particular for pointing out a (now corrected) issue in the proof of Lemma~\ref{lem:complexity}.

\bibliographystyle{IEEEtran}
\bibliography{ref}

\arxivver{}{

\appendix

\subsection{Proof of Lemma~\ref{lem:closed}}
Let $( w_k, p_k)_{k\in\N} \subset \B$ be a sequence of trajectories and scheduling signals such that $(w_k, p_k) \rightarrow (\hat w,\hat p)$ in $L^2(\intv, \R^\dnw) \times L^\infty(\intv, \R^{\dnp})$. We show that $(\hat w,\hat p)\in\B$. Consider any $\phi\in\mathcal C_{\mathrm c}^\infty(\intv,\R^\dnk)$. Then, the left-hand side of~\eqref{eq:weak} defines a continuous (nonlinear) functional $f_\phi:L^2(\intv, \R^\dnw) \times L^\infty(\intv, \R^{\dnp})\rightarrow \R$, $(w,p) \mapsto f_{\phi}(w,p)$. By continuity, we have that $f_\phi(\hat w,\hat p) = \lim_{k\rightarrow\infty} f_\phi(w_k, p_k) = 0$. Since $\phi$ was arbitrary, we conclude that $(\hat w,\hat p)\in\B$. \qed 

\subsection{Proof of Lemma~\ref{lem:embedding}}
Let $(w,p)\in\B$. Note that $w'=\col(w,p\otimes w)\in L^2(\intv, \R^{\dnw(1+\dnp)})$. Let $\phi$ be any test function in $\mathcal C_{\mathrm c}^\infty(\intv,\R^\dnk)$. Then, \eqref{eq:weak} holds and implies \eqref{eq:weak2}. This yields $w' \in \B'$.%
   
Conversely, let $w'=\col(w,v)\in \B'$ with $w\in L^2(\intv,R^{\dnw})$ and $v=p\otimes w$ for some $p\in L^\infty(\intv, \R^{\dnp})$. Then, \eqref{eq:weak2} holds for all $\phi\in\mathcal C_{\mathrm c}^\infty(\intv,\R^\dnk)$ and, hence, \eqref{eq:weak} holds. Therefore, $(w,p)\in\B$. \qed 

\subsection{Proof of Lemma~\ref{lem:complexity}}\label{app:proof:lem:complexity}
Before we give the formal proof of Lemma~\ref{lem:complexity}, we need a \emph{direct realization} of the input--output form~\eqref{eq:io} as a LPV state-space representation. We achieve this through the realization of~\eqref{eq:io} as a minimal LPV state-space realization with static scheduling dependence, for which we use the cut-and-shift method~\cite{rolandbook}. To smoothen notation, we will exclusively use the indeterminant~$\xi$, associated with~$\tfrac{\diff}{\diff t}$, and reverse operation~$\xi^{-1}$, associated with integration over time. %

Integrate both sides of~\eqref{eq:io} $\dnr$ times, which under Assumption~\ref{ass:coprime-etc} %
gives:
\begin{multline}\textstyle
    y  %
    + \textstyle\sum_{i=0}^{\dnr-1}A_{i,0}\xi^{i-\dnr} y + \sum_{j=1}^{\dnp}A_{i,j}\xi^{i-\dnr}(p_jy) \\  
    = \textstyle\sum_{i=0}^{\dnr-1}B_{i,0}\xi^{i-\dnr}u + \sum_{j=1}^{\dnp}B_{i,j}\xi^{i-\dnr}(p_ju) \\
    + \textstyle B_{\dnr,0}u + \sum_{j=1}^{\dnp}B_{\dnr,j} p_j u. 
\end{multline}
Inspired by~\cite{toth2011state}, define $x_1\in(\R^{\dny})^\mb{T}$ as:
\begin{multline}
     \textstyle x_1:=\sum_{i=0}^{\dnr-1}B_{i,0}\xi^{i-\dnr}u + \sum_{j=1}^{\dnp}B_{i,j}\xi^{i-\dnr}(p_ju) \\
     \textstyle - \sum_{i=0}^{\dnr-1}\left(A_{i,0}\xi^{i-\dnr} y + \sum_{j=1}^{\dnp}A_{i,j}\xi^{i-\dnr}(p_jy)\right),
\end{multline}
such that 
\begin{align}
    y & = x_1 + B_{\dnr,0}u + {\textstyle\sum_{j=1}^{\dnp}B_{\dnr,j}} p_j u %
    \label{eq:outputequationstate}\\ 
    \xi x_1  & = x_2 - A_{\dnr-1,0:\dnp}\begin{bsmallmatrix} 1 \\ p \end{bsmallmatrix} \otimes y + B_{\dnr-1,0:\dnp}\begin{bsmallmatrix} 1 \\ p \end{bsmallmatrix} \otimes u,
\end{align}
with $A_{\dnr-1,0:\dnp} = \begin{bmatrix} A_{\dnr-1,0} & \cdots & A_{\dnr-1,\dnp} \end{bmatrix}$, similarly for $B_{\dnr-1,0:\dnp}$, and 
\begin{multline*}
    \textstyle x_2:=\sum_{i=0}^{\dnr-2}B_{i,0}\xi^{i-\dnr+1}u + \sum_{j=1}^{\dnp}B_{i,j}\xi^{i-\dnr+1}(p_ju) \\
    \textstyle - \sum_{i=0}^{\dnr-2}\left(A_{i,0}\xi^{i-\dnr+1} y + \sum_{j=1}^{\dnp}A_{i,j}\xi^{i-\dnr+1}(p_jy)\right).
\end{multline*}
Repeating this \emph{cut-and-shift} method in the continuous-time case results until $\xi x_{\dnr}=B_{0,:}\begin{bsmallmatrix} 1 \\ p \end{bsmallmatrix} \otimes u-A_{0,:}\begin{bsmallmatrix} 1 \\ p \end{bsmallmatrix} \otimes y $. This construction yields the first order differential equation in~\eqref{eq:odess}. 
We can then partially substitute~$y$ as defined in~\eqref{eq:outputequationstate} in the equations for $\xi x_i$ in~\eqref{eq:odess} for only the outputs that are \emph{not} multiplied with the scheduling signal~$p$, which results in the LPV state-space realization with \emph{static} scheduling dependence shown in~\eqref{eq:odess2}. 
\begin{figure*}%
    \begin{subequations}\label{eq:odess}
    \begin{align}
    \xi \begin{bmatrix} x_1 \\ x_2 \\ \vdots \\ x_{\dnr-1}\\ x_{\dnr} \end{bmatrix} &= \begin{bmatrix} 0 & I & \cdots  & 0 & 0\\ 0 & 0 & I & \cdots & 0 \\ \vdots & \ddots & \ddots & \ddots & \vdots \\ 0 & \cdots  & \cdots & 0 & I \\ 0 & 0 & \cdots & \cdots & 0 
    \end{bmatrix}\begin{bmatrix} x_1 \\ x_2 \\ \vdots \\ x_{\dnr-1}\\ x_{\dnr} \end{bmatrix}   - \begin{bmatrix} A_{\dna-1,0} \\ A_{\dna-2,0} \\ \vdots \\ A_{1,0}\\ A_{0,0} \end{bmatrix} y   + 
    \begin{bmatrix}
    -A_{\dna-1,1:\dnp} & B_{\dna-1,0} & B_{\dna-1,1:\dnp} \\ 
-A_{\dna-2,1:\dnp} & B_{\dna-2,0} & B_{\dna-2,1:\dnp} \\
\vdots & \vdots &\vdots\\
-A_{1,     1:\dnp} & B_{1,0}&B_{1,1:\dnp} \\ 
-A_{0,     1:\dnp} & B_{0,0}&B_{0,1:\dnp} \end{bmatrix}\begin{bmatrix}
    p\otimes y \\ u \\ p\otimes u
\end{bmatrix}, \\
y & = x_1 + \begin{bmatrix} %
	B_{\dna,0} & B_{\dna,1:\dnp} 
\end{bmatrix}\begin{bmatrix} %
u \\ p\otimes u
\end{bmatrix}. 
\end{align}
\end{subequations}
\hrule

\smallskip

\begin{subequations}\label{eq:odess2}
\begin{align}
    \xi \begin{bmatrix} x_1 \\ x_2 \\ \vdots \\ x_{\dnr-1}\\ x_{\dnr} \end{bmatrix} &= \begin{bmatrix} -A_{\dna-1,0} & I & \cdots  & 0 & 0\\ -A_{\dna-2,0} & 0 & I & \cdots & 0 \\ \vdots & \ddots & \ddots & \ddots & \vdots \\ -A_{1,0} & \cdots  & \cdots & 0 & I \\ -A_{0,0} & 0 & \cdots & \cdots & 0 
    \end{bmatrix}\begin{bmatrix} x_1 \\ x_2 \\ \vdots \\ x_{\dnr-1}\\ x_{\dnr} \end{bmatrix}   + \begin{bmatrix}
    -A_{\dna-1,1:\dnp} & \hat B_{\dna-1,0} & \hat B_{\dna-1,1:\dnp} \\ 
-A_{\dna-2,1:\dnp} & \hat B_{\dna-2,0} & \hat B_{\dna-2,1:\dnp} \\
\vdots & \vdots & \vdots\\
-A_{1,     1:\dnp} & \hat B_{1,0}&\hat B_{1,1:\dnp} \\ 
-A_{0,     1:\dnp} & \hat B_{0,0}&\hat B_{0,1:\dnp} \end{bmatrix}\begin{bmatrix}
    p\otimes y \\ u \\ p\otimes u
\end{bmatrix}, \\
y & = x_1 + \begin{bmatrix} %
	B_{\dna,0} & B_{\dna,1:\dnp} \end{bmatrix}\begin{bmatrix} %
	u \\ p\otimes u
\end{bmatrix},\label{eq:odess2:output} \\
\text{where:} & \qquad  \hat B_{i,0} = B_{i,0} - A_{i,0}B_{\dnr,0}, \qquad \hat B_{i,1:\dnp} = B_{i,1:\dnp} - A_{i,0}B_{\dnr,1:\dnp} \notag
\end{align}
\end{subequations}
\hrule
\end{figure*}
{Note that the state equation of the realization can be made independent from the output by substituting~\eqref{eq:odess2:output} in the~$p\otimes y$ term, yielding an LPV state-space representation of the form $\xi x = A(p)x + B(p)u$; $y = C(p)x + D(p)u$, affine in~$p$. The state-space realization can be made minimal using an LPV Kalman decomposition~\cite{petreczky2016realization}. Importantly, the LPV Kalman decomposition uses a \emph{constant} projection matrix that projects the state to a lower dimension, i.e., the LPV state-space representation retains the affine dependency structure.} It follows from Assumption~\ref{ass:coprime-etc} and the proof of~\cite[Lem.~1, Lem.~3]{verhoek2026behavioral} that the rank of the observability matrix of~\eqref{eq:odess2} is scheduling-independent. Hence, the observability matrix of~\eqref{eq:odess2} will achieve rank $\order(\B)$ after $\lag(\B)$ steps, \emph{for any $p$---including $p\equiv0$}. {Thus,~\eqref{eq:odess2} is observable in the sense of~\cite{petreczky2016realization, petreczky2023minimal}. Then, by~\cite[Thm.~2]{petreczky2023minimal}, the (state-minimal) LPV Kalman decomposition of~\eqref{eq:odess2} is a minimal realization of the behavior~$\B$, i.e., it has a state dimension of~$\order(\B)$.} %
In the remainder, we assume w.l.o.g. that the state dimension of~\eqref{eq:odess2} is~$\order(\B)$, meaning that it is complete state-observable~\cite[Chap.~3,4]{rolandbook}. We are now ready to give the proof for Lemma~\ref{lem:complexity}:

\begin{proof}
  We first argue that~$\inputdim(\B) + \dnp \dnw=\inputdim(\B')$. Under the assumption that the signal $p\otimes w$ is free, we directly see that $\inputdim(\B')=\inputdim(\B) + \dnp \dnw$. {Moreover, as we argued before, the observability matrix of~\eqref{eq:odess2} achieves rank $\order(\B)$ after $\lag(\B)$ steps for any $p$. Hence, for a free~$p\otimes w$, the observability matrix of the LTI state-space realization resulting from~\eqref{eq:odess2} achieves rank $\order(\B)$ after $\lag(\B)$. Therefore, we can conclude $\order(\B)=\order(\B')$ and $\lag(\B) = \lag(\B')$.}
\end{proof}

\subsection{Proof of Proposition~\ref{prop:FL_LTI}}\label{app:proof:prop:FL_LTI}
   We consider a minimal input--state--output representation of the LTI behavior $\B'$ of the form
   \begin{equation}\label{eq:pf:ltiss}
   x^{(1)} = Ax + Bu',\quad y' = Cx + Du',
   \end{equation}
   such that $w' = \col(y',u')$ for $w'\in\B'$. {Minimality implies that the dimension of the state variable $x$ equals $\order(\B)$.} Let
   \begin{equation}
      \mathcal O_{k} = \begin{bmatrix}
         C\\
         CA\\
         \vdots\\
         CA^{k-1}
      \end{bmatrix},\quad \mathcal T_k = \begin{cases}
         D, & k=1\\
         \begin{bmatrix}
            D\\
            \mathcal O_{k-1}B & \mathcal T_{k-1}
         \end{bmatrix}, & k\geq 2,
      \end{cases}
   \end{equation}
   where $\mathcal O_k$ is a Kalman observability matrix and the block-Toeplitz matrix $\mathcal T_k$ contains the Markov parameters. Define
   \begin{equation}
      \mathcal S_L = \begin{bmatrix}
         \mathcal O_L & \mathcal T_L\\
         0 & I_{L \inputdim(\B')}
      \end{bmatrix}.
   \end{equation}
   {As~\eqref{eq:pf:ltiss} is observable (see \cite[Sec.~6.5.3]{yellowbook}) and $L\geq\lag(\B)+1$,} the matrix $\mathcal S_L\in\R^{L \dnwp \times (\order(\B')+L\inputdim(\B'))}$ has full column rank. Let $x$ be the state trajectory corresponding to $w'=\col(y',u')$. Then, it is easy to see that the symmetric, positive semi-definite  Gramian
   \begin{equation}
      \Gamma_0 = {\textstyle\int_\intv} \mc{W}_0(t)\mc{W}_0(t)^\top\,\diff t,%
   \end{equation}
   with $\mc{W}_0 =\col(x, u, \dots, u^{(L-1)})$, 
   satisfies 
   \(
      \Omega\,\Gamma = \mathcal S_L \Gamma_0 \mathcal S_L^\top,
   \)
   where $\Omega$ is an invertible permutation matrix that separates (the derivatives of)~$w'$ in terms of (the derivatives of) inputs and outputs, i.e., 
   \[ \Omega \col\begin{psmallmatrix}w',\dots,(w')^{(L-1)}\end{psmallmatrix} = \col\begin{psmallmatrix}y',\dots, (y')^{(L-1)}, u',\dots, (u')^{(L-1)}\end{psmallmatrix}. \]
   First, we show that (i) is equivalent to the statement
   \begin{itemize}
       \item[] (iii) $\operatorname{rank}\Gamma_0=L\inputdim(\B')+\order(\B')$.
   \end{itemize}
   Assume (i) holds, i.e., $\rank(\Gamma) = L\inputdim(\B') + \order(\B')$. Then,
   \begin{equation*}
      \rank(\Gamma) = \rank(\Omega\,\Gamma) = \rank(\mathcal S_L \Gamma_0 \mathcal S_L^\top) \leq \rank(\Gamma_0).
   \end{equation*}
   Since, $\Gamma_0\in\R^{(\order(\B')+L\inputdim(\B')) \times (\order(\B')+L\inputdim(\B'))}$ the statement~(iii) directly follows. Assume now that (iii) holds, i.e., $\Gamma_0$ is positive-definite. Then, $\operatorname{ker}(\Gamma) = \operatorname{ker}(\mathcal S_L\Gamma_0\mathcal S_L^\top) = \operatorname{ker}(\mathcal S_L)$. Taking the orthogonal complements yields $\operatorname{im}(\Gamma) = \operatorname{im}(\mathcal S_L^\top)$ and, thus, $\operatorname{rank}(\Gamma)=\operatorname{rank}(\mathcal S_L) = L\inputdim(\B')+\order(\B')$, i.e., statement~(i) holds.

   By adapting the proof of \cite[Thm.~22]{Schmitz2024a}, where the persistency of excitation condition is directly replaced with condition (iii), one sees that (iii) implies~(ii).

    Finally, we show that (ii) implies (i). Consider an arbitrary state--input trajectory $(\bar x,\bar u')\in W^{L-1,\infty}(\intv,\R^{\order(\B')})\times W^{L-1,\infty}(\intv,\R^{\inputdim(\B')})$ from the state-space representation~\eqref{eq:pf:ltiss}. The corresponding input--output trajectory $\bar w'=\col(\bar y',\bar u')$ is an element of $\B'\cap W^{L-1,\infty}(\intv,\R^{\dnwp})$ and satisfies
    \begin{equation}
    \label{eq:wuyx}
        \Omega\begin{bmatrix}
            \bar w'\\
            \vdots\\
            (\bar w')^{(L-1)}
        \end{bmatrix} = \begin{bmatrix}
            \bar y'\\\vdots\\(\bar y')^{(L-1)}\\\bar u'\\\vdots\\(\bar u')^{(L-1)}
        \end{bmatrix} = \mathcal S_L \begin{bmatrix}
            \bar x\\\bar u'\\\vdots\\(\bar u')^{(L-1)}\end{bmatrix}.
    \end{equation}
    Since, $\col(\bar x, \bar u',\dots, (\bar u')^{(L-1)})$ has $L\inputdim(\B')+\order(\B')$ degrees of freedom and $\mathcal S_L$ is injective, together with \eqref{eq:wuyx} the linear subspace of $\R^{L\dnwp}$
    \begin{equation*}
        \mathcal M =\left\{\begin{bsmallmatrix}
            \bar w'(t)\\
            \vdots\\
            (\bar w')^{(L-1)}(t)
        \end{bsmallmatrix}\,\middle|\,\begin{aligned}
        &\bar w' \in \B'\cap W^{L-1,\infty}(\intv,\R^{\dnwp}),\\ &t\in\intv\end{aligned}
        \right\}
    \end{equation*}
    has dimension $L\inputdim(\B')+\order(\B')$.
    Then, by statement~(ii), $\mathcal M\subset \operatorname{im}(\Gamma)$, which yields $\operatorname{rank}(\Gamma) \geq L\inputdim(\B')+\order(\B')$. By $\Omega \Gamma = \mathcal S_L\Gamma_0\mathcal S_L^\top$, we find
    $\operatorname{rank}(\Gamma) =\operatorname{rank}(\Omega\Gamma)\leq \operatorname{rank}(\mathcal S_L)= L\inputdim(\B')+\order(\B')$, which implies (i). 
    
    In summary, we have show the equivalence of the statements (i), (ii), and (iii), which completes the proof.\qed

\subsection{Proof of Theorem~\ref{thm:FL_LPV}}
   We have that $w'=\col(w,p\otimes w)\in\B'$ and, by Lemma~\ref{lem:complexity}, the Gramian $\Gamma$ in~\eqref{eq:gamma} satisfies $\rank(\Gamma) = L\inputdim(\B')  + \order(\B')$ with $L\geq \lag(\B') + 1$. Let $(\bar w,\bar p)\in W^{L-1,2}(\intv, \R^{\dnw})\times W^{L-1,\infty}(\intv,\R^{\dnp})$ and set $\bar w' = \col(\bar w, \bar p\otimes \bar w)$.
   
   First, we show that $(\bar w,\bar p)\in\B$ implies \eqref{eq:FL_LPV}. Let $(\bar w,\bar p)\in\B$. By Lemma~\ref{lem:embedding} we have $\bar w'\in\B'$ and by Proposition~\ref{prop:FL_LTI} there exists $g\in L^2(\intv, \R^{L\dnw(1+\dnp)})$ such that \eqref{eq:FL_LTI} holds. Since $(\bar w')^{(k)}=\col(\bar w^{(k)}, (\bar p\otimes \bar w)^{(k)})$ for $k=0,\dots, L-1$, we have~\eqref{eq:FL_LPV}.

   Now, we show that \eqref{eq:FL_LPV} implies $(\bar w,\bar p)\in\B$. Let \eqref{eq:FL_LPV} hold for some $g\in L^2(\intv, R^{L\dnw(1+\dnp)})$. Then \eqref{eq:FL_LPV} yields \eqref{eq:FL_LTI} and, hence, $\bar w'\in\B'$ by Proposition~\ref{prop:FL_LTI}. By Lemma~\ref{lem:embedding}, we have that $(\bar w,\bar p)\in\B$. {Thus, $(\bar w,\bar p)\in\B$ if and only if \eqref{eq:FL_LPV} holds}. \qed

\subsection{Proof of Corollary~\ref{cor:ofthmfl_lpv}}
    It suffices to show the equivalence of \eqref{eq:FL_LPV} to the set of equations \eqref{eq:rep}--\eqref{eq:compat}. First, we show that \eqref{eq:FL_LPV} implies \eqref{eq:rep}--\eqref{eq:compat}. Given \eqref{eq:FL_LPV}, we find
    \begin{align}
    \label{eq:DAE2}
        \begin{bmatrix}
            \bar w\\\vdots\\\bar w^{(L-1)}
        \end{bmatrix} &= \begin{bmatrix}
        \Gamma_{w}\\\vdots\\\Gamma_{w^{(L-1)}}
        \end{bmatrix} g,\\ 
        \label{eq:compat2}
        \begin{bmatrix}
            \bar p \otimes \bar w \\\vdots\\ (\bar p\otimes \bar w)^{(L-1)}
        \end{bmatrix} &= \begin{bmatrix} \Gamma_{pw}\\\vdots\\\Gamma_{(pw)^{(L-1)}} 
        \end{bmatrix} g.
    \end{align}
    This directly yields \eqref{eq:rep}. From~\eqref{eq:DAE2}, we have
    \begin{equation}
        \frac{\diff}{\diff t} (\Gamma_{w^{(j)}} g) = \bar w^{(j+1)} = \Gamma_{w^{(j+1)}} g
    \end{equation}
    for all $j=0,\dots, L-2$, where $\Gamma_{w^{(0)}}=\Gamma_w$, which directly implies~\eqref{eq:DAE}. Moreover,
    using the generalized Leibniz rule, we find
    \begin{equation}
        (\bar p\otimes \bar w)^{(n)} = \sum_{k=0}^n \binom{n}{k} \left( \bar p^{(k)} \otimes I_{\dnw}\right) \bar w^{(n-k)}
    \end{equation}
    for $n=0,\dots, L-1$. Consequently, the left-hand side in~\eqref{eq:compat2} can be expressed as
    \begin{equation*}
        \begin{bmatrix}
            \bar p \otimes \bar w \\\vdots\\ (\bar p\otimes \bar w)^{(L-1)}
        \end{bmatrix} = \bar{\mathscr P_L} \begin{bmatrix}
            \bar w\\\vdots\\\bar w^{(L-1)}
        \end{bmatrix} = \bar{\mathscr P_L}\begin{bmatrix}
        \Gamma_{w}\\\vdots\\\Gamma_{w^{(L-1)}}
        \end{bmatrix} g,
    \end{equation*}
    where we substituted~\eqref{eq:DAE2} to obtain~\eqref{eq:compat}. The reverse implication that \eqref{eq:rep}--\eqref{eq:compat} implies \eqref{eq:FL_LPV} follows directly from the above derivations. \qed
}
\end{document}